\documentclass[11pt]{article}

\usepackage{euscript}

\usepackage{mathrsfs}

\usepackage{graphics}

\usepackage {epic}

\usepackage {eepic}

\usepackage{latexsym}

\usepackage {epic}
\usepackage{amsfonts}
\usepackage{amssymb}
\usepackage{bbm}
\usepackage{dsfont}
\usepackage{euscript}
\usepackage{mathrsfs}
\usepackage{graphics}
\usepackage{graphicx}
\usepackage {epic}
\usepackage {eepic}
\usepackage{latexsym}
\usepackage{fancybox}
\usepackage{vmargin}

\setmarginsrb{2.5cm}{1.9cm}{2.6cm}{2.5cm}{0.5cm}{0.5cm}{0.5cm}{0.5cm}

\newtheorem{theorem}{Theorem}
\newtheorem{lemma}{Lemma}
\newtheorem{corollary}{Corollary}

\newenvironment{proof}{\noindent \textbf{Proof: }\ignorespaces}
  {\hspace*{\fill}$\Box$\medskip}

\title{Parameterized Algorithms for Clustering PPI Networks}

\author{
Sriganesh Srihari and Hon Wai Leong\thanks{School of Computing, National University of Singapore, 13 Computing Drive, Singapore 117590. 
E-mail: \{srigsri, leonghw\}@comp.nus.edu.sg.}}

\begin{document}
\date{}
\maketitle

\def\comment#1{{\footnotesize\ligne \vspace{-5pt} \noindent#1\vspace{-5pt}\\\ligne}}
\def\comment#1{}

\begin{abstract}
With the advent of high-throughput wet lab technlogies the amount of protein interaction data available publicly has increased
substantially, in turn spurring a plethora of computational methods for~{\em in silico} knowledge discovery from this data. In this paper,
we focus on~{\em parameterized} methods for modeling and solving complex computational problems encountered in such knowledge discovery from 
protein data. Specifically, we concentrate on three relevant problems today in proteomics, namely detection of lethal proteins, functional
modules and alignments from protein interaction networks. We propose novel graph theoretic models for these problems 
and devise practical parameterized algorithms. At a broader level, we demonstrate how these methods can be 
viable alternatives for the several heurestic, randomized, approximation and sub-optimal methods 
by arriving at parameterized yet optimal solutions for these problems. We substantiate these theoretical results by experimenting on
real protein interaction data of~{\em S.cerevisiae}(budding yeast) and verifying the results
using gene ontology.
\end{abstract}

\section{Introduction}
With the advent of many high-throughput wet lab technologies the amount of biological data available for analysis
has increased substantially. For instance, Mass Spectrometry and Tandem Afinity Purification(MS-TAP), Yeast to Hybrid(Y2H) 
and ChIP-on-Chip experiments have contributed to large publicly accessible biological databases like BIND,
Biogrid and MIPS~\cite{Databases}. 
This in turn has increased the need to make sense out of such large quantities of data and to extract useful and intelligent information
from them. For instance, discovery of disease-causing genetic information from the population in a certain
region, protein structural information from protein databases, phylogenetic distances from genomic data, etc. 
This can have applications in areas like drug discovery, population genetics and phylogenetics.

However, in most cases, extraction of such intelligent and useful information is not easy. 
Many of the problems encountered are NP-hard, with researchers working on
sub-optimal, heurestic-based, yet fast solutions for solving these problems. 
Such {\em in silico} knowledge discovery from biological data is a very complicated area of research. 
The techniques adopted span several areas of computing like algorithmic theory, statistics, machine learning, etc. 
However, no particular technique
or set of algorithms can claim to solve a problem completely. Even if a particular technique solves a problem
efficiently and exactly, the solution may turn out to be not useful for biologists. Biologists not always look
for exact or fast solutions or complicated methods. Many times the `biological relevance' of the solution is very important. Hence, searching for
solutions in computational biology 
is usually a subtle balance between specificity(measured in terms of false positives and false negatives), efficiency(speed) and biological relevance. 
This is one the reasons why
many techniques developed these days are~{\em ensemble} methods that mix and mash several methods each strong in one sense.

In this paper, we intend to show the recent area of~{\em Parameterized Complexity} as an alternative to existing
techniques used to solve hard problems in computational biology. In order to show this, we take up several relevant problems motivated from
computational proteomics and~{\em model}~them as interesting theoretical problems, thereby proposing parameterized algorithms for the same. 
In this manner, on one hand, we deal with interesting
theoretical results which have relevance to computer science and on the other, concrete applications to computational biology.
We do not claim that our techniques are most efficient
or give best biologically relevant solutions, but they certainly are worthy alternatives to several existing techniques, and hence can also
be adopted by ensemble methods.

\paragraph{Problems addressed and related work:}
In this paper, we specifically focus on knowledge extraction from~{\em protein interaction} data. 
We try to solve two problems,
\vspace{0.1cm}

(1)~extraction of hubs and functional modules from protein interaction networks, and 

(2)~alignment of protein interaction networks across species. 

Some previous works in this area have focused on heurestic methods~\cite{Guangyu},
linear programming and matrix methods~\cite{Chris_Ding} and clustering~\cite{Chuan}. 
Kumar~\cite{Vipin} gives a detailed survey on data mining based methods.
Some applications of parameterized algorithms in computational biology 
can be found in the works of B\"{o}cker~\cite{Bocker1,Bocker2},
Neidermeier~\cite{Nie06}, Langston~\cite{Langston} and some of our preliminary work on hubs 
and quasi cliques in~\cite{Sri_HonWai}. Having given the references, 
we now begin our work with a brief background of the problems in focus:

\paragraph{Protein networks:}
Proteins are organic compounds that form vital constituents of living organisms. They are responsible 
for several biological processes and pathways. It has been observed that proteins interact with one another while performing functions. 
These~{\em interactions} are captured in the form of~{\em Protein-Protein Interaction} or PPI networks that are now publicly 
available in databases such as BIND, MIPS and Biogrid~\cite{Databases}. PPI networks demonstrate the so-called~{\em scale-free} properties. 

\paragraph{Scale-free properties of protein networks:}
Though not all biological networks demonstrate scale-free properties, it is believed 
that proteins networks(PPI) display most of these~\cite{Alexei,Wutchy}. Scale-free networks are those that 
follow power-law distribution in connectivies, that is, 
probability that a node links to $k$ other nodes is $P(k)=k^{-\alpha}$, where $\alpha$ is a small constant~\cite{Barabasi}.
This leads to a skewed distribution - there exist a small set nodes with very large interactions while others have 
far less interactions. Such a phenomenon gives
rise to two interesting structures namely, `hubs' and `communities'. Hubs in protein networks represent~{\em lethal} or important proteins
that interact with most other proteins and hence hold the network intact. Most lethal proteins are strategically located
and their disruption could possibly lead to biological lethality. Hence, their study is important to understand the causes of diseases.
For example, the deletion of the myosin pair~{\sc myo3-myo5} causes severe defects in growth and cytoskeleton organization~\cite{NCBI}.
Communities refer to~{\em functional modules}. Proteins within a community form a densely connected region within the network
and perform specific biological functions. For instance, functional categories of proteins involved in stress response, 
biosynthesis of vitamins and prosthetic groups~\cite{Guangyu}.
And therefore detecting communities is relevant to understand protein affinity in biological processes.

\paragraph{Protein network alignment:}
Researchers usually derive phylogenetic or evolutionary distances between species by comparing 
their genome sequences - calculating edit distances, reversals, etc. However, some recent works~\cite{Recomb08} have shown 
that alignment of protein networks can also lead to useful clues in constructing phylogenetic trees. 
Alignment here refers to deriving local isomorphisms between protein networks of two or more species to
understand the similarity traits between them; analogous to sequence alignments. Protein network alignments can be useful
to derive network~{\em motifs}(frequent local network patterns) across protein networks of species~\cite{Berg}.

\section{Preliminaries}
\label{Prelims}
{\em Parameterized Complexity} can be considered as a two-dimensional generalization of the classical complexity
theory, with one dimension being the size of the input instance, $|I|=n$ and second being a fixed input parameter $k$~\cite{DF99,FG06}.
Parameterized complexity studies the hardness of problems based on these two dimensions and does a more finer classification
of problems into: $FPT < W[1] < W[2]...$ Among these classes is the class of FPT or~{\em Fixed-parameter Tractable} problems 
which are interesting because they may be practically solvable. Formally, a problem is in the class FPT if it can be solved in time $O(f(k)n^{O(1)})$, 
where $f$ is a function of the fixed parameter $k$ alone~\cite{Nie06}. If under certain practical conditions, $k$ turns out to be `small' compared to $n$,
then these solutions turn out to be more tractable compared to the trivial $O(n^k)$ solutions for such problems.
Even more importantly, FPT solutions arrive at optimal solutions (subject to parameterization) and hence form interesting
alternatives to approximate, randomized or heurestic approaches. And hence, it is worthwhile exploring parameterized
alternatives to problems in computational biology because many times missing out some results due to approximation
may lead to missing out on biologically relevant cases.

Throughout this paper, we model protein networks as simple undirected graphs - proteins as nodes and their interactions as edges. 
Given a graph $G=(V,E)$, $n$ represents number of vertices,
and $m$ represents the number of edges. For a subset $V'\subseteq V$, by $G[V']$ we mean the subgraph 
of $G$ induced on $V'$. By $N(v)$ we represent all vertices (excluding $v$) that are adjacent to $v$, and by $N[v]$, 
we refer to $N(v)\cup \{v\}$. The degree of $v$ is $d(v)=|N(v)|$.

\section{Modeling hubs and functional modules in protein networks}
\label{modeling}
We first deal with the problem of extracting hubs and functional modules.
Most previous approaches have looked at the problem of detecting functional modules separately from that of hub proteins.
However, hub proteins form the key constituents of most functional modules. In any functional module, the hub proteins not only hold 
the other proteins within that module together, but also determine the most dominant functions of that module. 
Also, protein networks are usually very sparse with several singleton(isolated) and loosely-connected proteins.
Such proteins may be connected to one or more hub or non-hub proteins,
but generally do not take part in any functional module. Taking such proteins into consideration induces false positives, thereby
affecting the specificity of the results.
As a result of these observations, we approach the whole problem in the following three steps:
\vspace{0.1cm}

{\bf Step 1:}~Determine all the hub proteins within the network,

{\bf Step 2}:~Determine all the proteins(hubs and non-hubs) which are more likely to 
take part in some functional module(s) and filter away the rest, 

{\bf Step 3:}~Combine the information from above steps to determine the individual functional modules.
\vspace{0.1cm}

We now describe these steps in the following paragraphs:

\paragraph{Theoretical modeling:}
The main theoretical counterpart of our practical problem in hand is what we call the~{\sc list dominating set} problem:
\vspace{0.1cm}

{\sc list dominating set:}
For a given graph $G=(V,E)$, an integral list $L=\{l(v_1),\cdots,l(v_n)\}$,
$0<l(v_i) \leq d(v_i)$, $1 \leq i \leq n$, and a positive integer $k$, a subset $D \subseteq V$ 
is called~{\em $L$-dominating}~if for every $v_i\in V$, either $v_i\in D$ or $|N(v_i)\cap D|\geq l(v_i)$. 
The problem asks if we can find a $D$ such that $|D|\leq k$?
\vspace{0.1cm}

This was first introduced in~\cite{har_ger} under the name of `vector dominating set'.
The `vector' $L$ gives the number of neighbors any vertex $v \in V\setminus D$ needs to have in $D$.
This is an interesting problem which generalizes the classical~{\sc vertex cover} 
problem when $l(v_i)=d(v_i)$, and the~{\sc dominating set} problem when $l(v_i)=1$. It is NP-complete
in the general complexity domain, but in the parameterized domain, it scales a whole `tower' of
complexities from FPT(vertex cover) to W[2]-hard(dominating set), and the `jump' has been studied
by us in~\cite{Raman_Saket_Sri}.

In this paper, we show how this~{\em single} problem can serve our~{\em dual} purpose of modeling hubs and participating
proteins~(both steps 1 and 2) very well. However, since in general it is W[2]-complete~\cite{Raman_Saket_Sri},  
we require a suitable reformulation of the problem to arrive at a practical solution, which we propose as follows:
\vspace{0.1cm}

{\sc $(k,\lambda)$-list dominating set:} For a given graph $G$, an integral list $S=\{s(v_1), \cdots, s(v_n) \}$,
$0 \leq s(v_i) < d(v_i)$, $1 \leq i \leq n$ and positive parameters $k$ and $\lambda$, 
a subset $D \subseteq V$ is called $(k,\lambda)$-dominating 
if for every $v_i \in V$, either $v_i \in D$ or 
$|N(v_i) \cap D| \geq d(v_i)-s(v_i)$.
The problem asks whether we can find a $D$, such that $|D| \leq k$ and
$\Sigma_{v_i \in V \setminus D}{s(v_i)} \leq \lambda$?
\vspace{0.1cm}

Here, we refer to $S$ as `slack' list, given by $S=\{s(v_1),\cdots, s(v_n) \} =
\{d(v_1)-l(v_1),\cdots,d(v_n)-l(v_n)\}$. For a vertex $v_i$ outside $D$, $d(v_i)-s(v_i)$ gives the
number of neighbors of $v_i$ in $D$ and $s(v_i)$ gives the number of neighbors outside $D$, while
$\lambda$ is an~{\em additional} parameter introduced that bounds the number of~{\em edges} allowable between the vertices outside $D$.
Note that if $\Sigma_{v_i \in V \setminus D}{s(v_i)} \leq \lambda$, at most $\frac{\lambda}{2}$ edges are allowable between
the vertices outside $D$.
\vspace{0.1cm}

Next, we show how this new formulation models the steps stated above very well:

\paragraph{Step 1 - Modeling hubs:}
Intuitively, hubs are~{\em high-degree} nodes with which most of other nodes in the network interact. 
Hubs have been previously modeled as vertex covers~\cite{Langston} or by just blindly picking all nodes with degrees
above a certain threshold. Modelling hubs by
vertex covers {\em forces} all linkages to be
{\em covered} while threshold-based approach
may miss some important hubs that have degrees below
the threshold. So, what is required is a suitable `tunable' parameter that ensures that there is
some `relaxation' in the modeling. 

Therefore, we model hubs by {\sc $(k,\lambda)$-lds} for which the $s(v_i)$ values 
are `small'. This results in set $D$ becoming the hub-set for the graph $G$ covering all
but $s(v_i)$ edges incident on $v_i \in V \setminus D$.
In the special case of $s(v_i)=0$, $D$ becomes the the vertex cover of $G$. 
However, the advantage~{\sc $(k,\lambda)$-lds} offers is the
flexibility by means of `tunability' of $S$.
For any vertex, we can control the number of neighbors inside and outside $D$ by tuning $S$.

Such a tunability models real-world networks like PPI very well.
For instance, any~{\em non-hub} protein can have few~{\em direct}
interactions with other non-hub proteins to achieve certain biological functionalities, 
but at the same time be {\em densely} linked to hub proteins~\cite{Sri_HonWai}.

\paragraph{Step 2 - Modeling participating proteins:}
Technically, we need to arrive at the main~{\em core} that contains the participating proteins(involves hubs and non-hubs)
of all the functional modules in the network.
We model this core as a relatively dense sub-network in the protein network. 
We essentially look for a single(possibly huge) such sub-network.
Protein networks being usually very sparse, the density of such a core may not be very high, 
but by finding one we look to filter away all the isolated and loosely-connected proteins. 
A core $C$ is associated with two parameters, $\gamma$, the edge density and $|C|$, the size of the core. The density $\gamma$
is defined as the ratio of the number of edges in $C$ to the total number of edges in a complete subgraph of the same size.


It is interesting to note that when we consider the~{\em complement} of a given graph as an instance of the~{\sc $(k,\lambda)$-lds} problem 
to find the set $D$, the remaining vertices~(in $V \setminus D$) form a $(n-k)$-size core $C$ of the original graph.
$s(v_i)$ gives the maximum number of edges missing from every vertex $v_i$ in the core. 
Since $\Sigma_{v_i}{s(v_i)} \leq \lambda$,
at most $\frac{\lambda}{2}$ edges can be missing from the core, which gives a lower bound on $\gamma$. Again $S$ acts as a
`tunable' vector to control the sensitivity of the core.
Hence, we can use our formulation to not only bound the number of missing edges but also specify the edges corresponding to which
vertices can be missing. 

\paragraph{Step 3 - Combining the information:}
The above steps make it easier to determine the individual functional modules.
One natural way to proceed is to consider each hub $h \in D$
and form the set $FM(h) = h \cup [N(h)\cap C]$, that is, $FM(h)$ is the set of all vertices in $C$ that are adjacent to hub $h$.
We then calculate the edge densities(or average degrees) of each set and rank them. We consider each set as a functional module
and verify it based on shared gene-protein ontologies.


\section{Modeling of disjoint functional modules}
The modeling described in the previous section has the flexibility of giving overlapping functional modules, that is, same proteins
participating in more than one module. However, 
many times one looks for disjoint or non-overlapping functional modules that are expected to perform different functions.
In this section, we show a model to determine such non-overlapping modules and to do so
we introduce the following abstract problem:
\vspace{0.1cm}

{\sc cluster editing:}
Given an input graph $G=(V,E)$ and positive integers $k$ and $\lambda$, this problem
asks whether we can modify $G$ to consist of disjoint clusters by adding or deleting at most $k$ edges
such that total sum of the edges missing across the resultant clusters is at most $\lambda$.

By clusters we mean~{\em partial cliques} which may have `few' edges missing compared to cliques
(complete subgraphs).
The~{\sc qce} problem models disjoint modules in a straightforward manner. After $k$ edge edit operations, if
we can obtain a solution, the resultant graph will be disjoint union of multiple clusters. In protein networks, 
each such cluster could represent a disjoint module. A similar modeling by Niedermeier et al.~\cite{Nie06} considers
cliques for each disjoint module, however cliques are too restrictive and do not cater to naturally 
occuring communities in networks.

\section{Modeling protein network alignments}
\label{alignment}
We next concentrate on the second problem we mentioned earlier namely, protein network alignments across species.
Modeling of protein network alignments leads us to local graph alignment, which we call~{\em quasi isomorphism}.
By quasi isomorphism we mean that given two
labeled graphs(labeled by protein annotations), we try to find local regions in the graphs that are `highly' similar or
`highly' isomorphic. This can be considered as a relaxation on the local graph isomorphism problem where one tries to find
exact one-to-one correspondence between the labeled nodes and edges in two regions.

In order to model quasi isomorphism, we first state the concept of~{\em product} graph of two graphs. Let $G_1(V_1,E_1)$ and
$G_2(V_2,E_2)$ be two given graphs. Let $l$ be any function that maps the labels of vertices of $G_1$ to $G_2$. We can
construct the product graph $H$ of $G_1$ and $G_2$ as follows:
\vspace{0.1cm}

(a)~In graph $G_1$, if there is an edge $e_1=(u_1,v_1) \in E_1$ between vertices labeled $u_1,v_1 \in V_1$
(or if there is no edge between vertices labeled $u_1,v_1 \in V_1$), and
in graph $G_2$, there is a corresponding edge $e_2=(u_2,v_2) \in E_2$ between vertices labeled $u_2=l(u_1),v_2=l(v_1) \in V_2$
(or if there is no edge between vertices labeled $u_2=l(u_1),v_2=l(v_1) \in V_2$), then
we add vertices labeled $\{u_1,v_1\}$ and $\{u_2,v_2\}$ to the vertex set of $H$ and add an edge between them.
\vspace{0.1cm}

(b)~In all other cases, we do not add any labeled vertices nor edges to $H$.
\vspace{0.1cm}

An interesting observation here is that a core $C$ in the product graph $H$ corresponds to local regions
in graphs $G_1$ and $G_2$ that are `highly' similar. That is,
\vspace{0.1cm}

(a)~the vertices of these two local regions display labeled mapping. If the region of $G_1$ has vertex $u_1$, then
the region of $G_2$ will have vertex $u_2=l(u_1)$, and

(b)~there is~{\em partial} correspondence between the edges in these two local regions, that is, there can be a bounded number of
{\em mismatches}. By mismatches we mean if there is an edge $(u_1,v_1)$ in the local region of $G_1$, then there is no
edge between $u_2=l(u_1)$ and $v_2=l(v_1)$ of $G_2$ or vice versa. 
These mismatches are bounded by the number of edges missing
in the core $C$ of the product graph $H$. 

To apply this to PPI networks, we first build the product graph $H$ based on protein labelings and then find the core $C$ in $H$.
Here, the mapping function $l$ can be an identity function $l(x)=x$ or can depend on the annotation schemes 
used(for example, the same protein could be labeled by different names in different species).

\section{Parameterized algorithms}
\label{solutions}
We dedicate this section to propose parameterized algorithms for all the theoretical problems introduced till now.

\subsection{Solution to {\sc $(k,\lambda)$-lds}} We begin by giving the following useful lemma:

\begin{lemma}
\label{lemma1}
If for $v_i \in V$, $d(v_i) > k + \frac{\lambda}{2}$, then $v_i$ is part of every
$(k,\lambda)$-dominating set $D$.
\end{lemma}

\begin{proof}
Observe that if $d(v_i) > k + \frac{\lambda}{2}$ and $v_i \in V\setminus D$, then $D$ will be able to accommodate
at most $k$ neighbors of $v$ forcing the remaining neighbors to be in $V \setminus D$. This would mean
$\Sigma_{v_i \in V\setminus D}{s(v_i)} > \lambda$. Hence, $v_i$ needs to be part of every 
$(k,\lambda)$-dominating set $D$.
\end{proof}

This leads to the following FPT algorithm for the {\sc $(k,\lambda)$-lds} problem:

\begin{theorem}
\label{theorem1}
{\sc $(k,\lambda)$-lds} is FPT in $k$ and $\lambda$.
\end{theorem}
\begin{proof}
Initially we set $U = V$ and $D = \emptyset$. And our parameters are $k > 0$ and $\lambda > 0$.
We use Lemma~\ref{lemma1} to arrive at data reduction techniques which, in parameterized complexity literature, are
called~{\em Kernelization Rules}:

If there is a vertex $v_i$ such that $d(v_i) > k + \frac{\lambda}{2}$, then 
do $D := D \cup \{v_i\}$ and $U := U \setminus \{v_i\}$.
We apply the rule exhaustively till there are no more vertices $v_i$. 
If the resultant solution set $D$ has more than $k$ vertices, then there exists no solution and we return~{\sc no}.

After application of the rule, all vertices $v_i \in U$ have $d(v_i) \leq k + \frac{\lambda}{2}$.
Therefore, the number of edges in the remaining induced subgraph $G[V \setminus D]$ is 
at most $(k-|D|)(k+\frac{\lambda}{2}) + \frac{\lambda}{2}$, {\em if it is to have a solution}.
The reason for this is: $D$ can accommodate $k-|D|$ more vertices, each of which has
degree at most $k+\frac{\lambda}{2}$. Hence, moving any vertex from $V\setminus D$ to $D$ can cover at most
$k+\frac{\lambda}{2}$ edges. Plus, allowable $\frac{\lambda}{2}$ edges in $G[V\setminus D]$.
If the number of edges is more than this, we return~{\sc no}. 
This ends our kernelization step.

Next, we then perform a~{\em Depth-bound Search} on the remaining graph.
We set our new parameter to $k' = k-|D|$.
At every step of the search we maintain two partitions of $U$:
\begin{itemize}
\setlength{\itemsep}{-0.15cm}
\item $X$: for every $v_i \in X$, $|N(v_i) \cap D| \geq d(v_i)-s(v_i)$,
\item $Y$: for every $v_i \in Y$, $|N(v_i) \cap D| < d(v_i)-s(v_i)$.
\end{itemize}
The partition $X$ consists of all vertices $v_i$ that have their required (at least $d(v_i)-s(v_i)$)
neighbors in $D$, while the partition $Y$ has all other vertices of $U$.
We refer the vertices in $X$ as~{\em saturated} and those in $Y$ as~{\em unsaturated}. 

We pick an edge $(u,v)$ from the remaining graph and branch upon the following conditions: 
either $(u,v)$ is in $G[V\setminus D]$ or it is not. If not, then either $u$ or $v$ is in $D$.  
So, we recursively solve the problem by performing a three-way branching at each step. 
More specifically, if we move a vertex $u$(or $v$) into $D$, we set $D := D \cup \{u\}$ (or $D:=D \cup \{v\}$) 
and $U:= U \setminus \{u\}$ (or $U := U \setminus \{v\}$). 
By doing this, if we can find vertices $w \in Y$ such that $|N(w) \cap D| \geq d(w)-s(w)$, 
then we move $w$ from $Y$ to $X$, by doing $X:= X \cup \{w\}$ and $Y:=Y-\{w\}$ (that is, $w$ is now
{\em saturated}). We reduce the parameter $k'$ by 1. 
However, if the edge $(u,v)$ is retained in $G[V\setminus D]$, we reduce the 
parameter $\frac{\lambda}{2}$ by 1.
At any particular step, if $k=0$ and $Y \neq \emptyset$,
we return~{\sc no}. Else, we return~{\sc yes} and the solution set $D$.

The correctness of the algorithm is clear from the description. 
For the time complexity, observe that at each recursive step, 
we perform a three-way branching and the depth of the recursion tree is 
at most $d=k+\frac{\lambda}{2}$. Hence, the total number of nodes in the tree is 
bounded by $3^{k + \frac{\lambda}{2}}$. Since we spend polynomial time for kernelization 
and at each step of the search tree, the complexity of our algorithm is 
bounded in the worst case by $O(3^{(k+\frac{\lambda}{2})}n^2)$, which is FPT in $k$ and $\lambda$.
\end{proof}

{\em Pruning the search tree:} This is based on the observation
that when we pick an edge $(u,v)$ on any recursive call, if $s(u)=0$ or $s(v)=0$ or $\lambda = 0$ 
then we can avoid the third branch altogether. Hence, the complexity can be reduced to $O(c^{(k+\frac{\lambda}{2})}n^2)$,
$2 \leq c < 3$.

We now derive a relationship between~{\sc $(k,\lambda)$-lds} and core through the following lemma,

\begin{lemma} 
For a graph $G=(V,E)$, let $S=\{s(v_1), \cdots, s(v_n) \}$ be an integral vector, $k$ and $\lambda$ be 
positive parameters. If $G$ has a $(k,\lambda)$-list dominating set $D$
of size at most $k$ such that $\Sigma_{v_i \in V\setminus D}{s(v_i)} \leq \lambda$, then the 
subgraph $\bar{G}[V\setminus D]$ of the complement
graph $\bar{G}$ of $G$, is a core $C$ of size at least $(n-k)$ with density
$\{ {n-k \choose 2} - \frac{\lambda}{2} \}/{n-k \choose 2} \leq \gamma \leq 1$.
\end{lemma}

\begin{proof}
If $D \subseteq V$ is the desired $(k,\lambda)$-lds of $G$ such that $|D| \leq k$ 
and $\Sigma_{v_i \in V\setminus D}{s(v_i)} \leq \lambda$,
then we can infer two points: the size of $V\setminus D$ is at least $(n-k)$ and every vertex $v_i \in V\setminus D$ 
has at most $\frac{\lambda}{2}$ neighbors in $V \setminus D$. 
In other words, the subgraph induced on $V\setminus D$ has at least $(n-k)$ vertices with at most $\frac{\lambda}{2}$
edges between the vertices. Hence, the complement of the subgraph will also be of size at least $(n-k)$, but will
have at most $\frac{\lambda}{2}$ edges missing. So, the total number of edges in the complement subgraph would
be at least $\{ {n-k \choose 2} - \frac{\lambda}{2} \}$. Hence, we get a core with density $\gamma$:
$\{ {n-k \choose 2} - \frac{\lambda}{2} \}/{n-k \choose 2} \leq \gamma \leq 1$.
\end{proof}

This results in the following corollary for core in the original graph,
\begin{corollary}
For a graph $G=(V,E)$, let $S=\{s(v_1),\cdots,s(v_n)\}$ be an integral vector, $k$ and $\lambda$ be positive
input parameters.
If the complement graph $\bar{G}$ has a $(k,\lambda)$-lds $D$ such that $|D| \leq k$ and
$\Sigma_{v_i}{s(v_i)} \leq \lambda$, then the subgraph $G[V\setminus D]$ of the original graph $G$
is a core $C$ of size at least $(n-k)$ with density $\{ {n-k \choose 2} - \frac{\lambda}{2} \} / {n-k \choose 2}
\leq \gamma \leq 1$.
\end{corollary}

\subsection{Solution to {\sc cluster editing}}

We give a search tree-based FPT algorithm for this problem.
We search for~{\em conflict triples} in the graph. A conflict triple $\{u,v,w\}$ has vertices $u$ and $v$ 
connected by an edge and $u$ is connected to a vertex $w$ that is not connected to $v$.
Such a conflict triple compels us to do one of four things: (a)~add an edge $(v,w)$, (b)~remove the edge $(u,v)$,
(c)~remove the edge $(u,w)$ (d)~allow the triple to be as it is.
For each of the first three options, we will require to do an edit operation, and hence counts to the
parameter $k$. The last option counts to the parameter $\lambda$. This is because, by allowing
the conflict triple to remain as it is, we are allowing the cluster containing this triple to have a missing
edge. By branching on each of these cases, we can arrive
at a $O(4^{k+\lambda}n^2)$ FPT algorithm.

One way to improve the algorithm would be to maintain markers~{\em permanent} and~{\em forbidden}.
When an edge $(u,v)$ is added, it is marked~{\em permanent} so that subsequent calls do not remove it.
Similarly, when an edge $(u,v)$ is removed, it is marked~{\em forbidden} so that subsequent calls do not add it.
With these markers, we can arrive at better worst-case bounds for the search tree.

\subsection{Solution to~{\sc quasi isomorphism}}

We state the following lemma and explain it.
\begin{lemma}
Let $H$ be the product graph of two given graphs $G_1$ and $G_2$. A core $C$ in $H$ with at most $\frac{\lambda}{2}$
edges missing corresponds to local quasi isomorphism between $G_1$ and $G_2$ with at most $\frac{\lambda}{2}$ mismatches.
\end{lemma}

Let $C$ be the core of $H$ obtained by the
algorithm in Theorem~\ref{theorem1}. If $C$ has $\frac{\lambda}{2}$ edges missing then these edges correspond to $\frac{\lambda}{2}$
mismatches between two local regions of $G_1$ and $G_2$. That is, if an edge between vertices labeled
$\{u_1,v_1\}$, $\{u_2,v_2\}$ of $H$ is missing in $C$, then either (a)~edge $(u_1,v_1) \in E_1$ (of $G_1$)
and edge $(u_2,v_2) \notin E_2$ (of $G_2$), or (b)~edge $(u_1,v_1) \notin E_1$ (of $G_1$)
and edge $(u_2,v_2) \in E_2$ (of $G_2$). The number of such mismatches is bounded by $\frac{\lambda}{2}$.
\vspace{0.1cm}

\section{Experiments and results}
\label{results}

\paragraph{Experimental setup:}
Here, we only show the experimental results of modeling hubs and functional modules using the~{\sc $(k,\lambda)$-lds problem}.
We implemented the algorithm of~{\bf Theorem 1} on an Intel Xeon 2.4GHz 4GB RAM Debian machine. 
The protein interaction dataset is of~{\em Saccharomyces Cerevisiae}
(budding yeast) from the Biogrid database~\cite{Databases}. 
After cleaning and removing self and multi edges, the resultant network had 1562 proteins and 1408 edges. 
The highest degree of the network was 46. About 5\% of proteins had degrees in the range 20-35, while rest were below 3. We left the 
isolated proteins to be filtered away by the algorithm instead of during cleaning.

{\em Setting the parameters:}~$k$ is the parameter that determines the size of the solution $D$. 
Through several experimental runs we have noticed that the algorithm
works very efficiently for scale-free networks and typically gives a solution for small(relative to $n$) values of $k$. 
This is mainly because of the presence of few high-degree nodes which when picked cover most of the interactions.
The complements of sparse scale-free networks turn out
to be dense and a small $k$ typically does not give a solution. 
If $k$ is very small, too many nodes will get included into the solution during kernelization, thereby overshooting the
solution size. If $k$ is too large, the problem may not kernelize well. But a large $k$ does not necessarily make the problem intractable 
in practice because the theoretical bound is essentially a worst case bound. 
$\lambda$ is set much smaller compared to the total number of edges. We chose to keep slack 
values of the vertices proportional to their degrees by fixing a division factor $r$ while setting vector $S$.

\paragraph{Detection of hubs and participating proteins:}
Some notations used to represent the results are:
$m'$ is number of edges in the complemented graph $\bar{G}$, $r$ is the division factor and the values of the vector $S$ 
are set as $s(v)=\frac{d(v)}{r}$ for all $v \in V$, 
$\gamma$ is the density of the core in the original graph $G$ and $T$ is the time in seconds. 

{\bf Table~\ref{table: PPI_hubs}} gives the results of detecting hub proteins on four runs with different parameter settings. The 
slack values for most vertices are usually very small because of the division factor $r$. 
The higher these slack values, the higher the number of slack edges allowable and hence, smaller the size of solution set $D$. 
Secondly, in practice the algorithm runs very fast eventhough the theoretical worst case bound is exponential.
{\bf Table~\ref{table: PPI_hubs_bio}} gives brief biological descriptions of few hubs that were discovered. 
The myosin pair~{\sc myo3-myo5} was present in our dataset
and the algorithm has successfully found it. 
Their complete biological descriptions
can be obtained by searching for the protein name in the Biogrid website~\cite{Databases}. 
The full list of hubs discovered can be got from~\cite{Datasets_Results}.
{\bf Table~\ref{table: PPI_qc}} gives the results of four runs on the complemented protein network. The points to notice here are that the complements
of protein networks are usually dense. The algorithm runs slower on these complemented networks and hence discovering a core takes longer time
than hubs. Also, the lesser the slack values in the complemented graph, the higher the density $\gamma$ of the core in the original graph. 
This is because of fewer slack edges allowable among vertices of set $C$ in complemented graph and hence, lesser the number of edges missing in the 
core of the original graph.

\begin{table}[ht]
\scriptsize
\centering
\begin{tabular}{c c | c c c | c c c}
\hline\hline
$n$ & $m$ & $k$ & $\lambda$ & $r$ & $|D|$ & $|V \setminus D|$ & $T$\\[0.5ex]
\hline
1562 & 1408 & 400 & 100 & 4  &  320 & 1242 & 1.00\\
1562 & 1408 & 400 & 200 & 6 &  341 & 1221 & 1.02\\
1562 & 1408 & 450 & 200 & 8 &  341 & 1221 & 1.06\\
1562 & 1408 & 450 & 200 & 10 &  346 & 1216 & 1.04\\
\hline
\end{tabular}
\caption{Hubs: four runs on yeast protein network}
\label{table: PPI_hubs}
\end{table}

\begin{table}[ht]
\scriptsize
\centering
\begin{tabular}{c | c c}
\hline\hline
Hub protein & Description & GO annotations\\[0.5ex]
\hline
MYO5 & deletion with MYO3 affects growth,cytoskeleton org & myosin binding;exocytosis\\
MYO3 & deletion with MY05 affects growth,cytoskeleton org & microfilament motor activation\\
RVS167 & actin-associated;endocytosis; & cytoskeleton protein binding;osmotic stress\\
VRP1 & cytokinesis;mammalian WASP syndrome & actin binding;cytoskeleton org\\
\hline
\end{tabular}
\caption{Descriptions of some hubs verified through the Biogrid~\cite{Databases}}
\label{table: PPI_hubs_bio}
\end{table}

\begin{table}[ht]
\scriptsize
\centering
\begin{tabular}{c c c | c c c | c c c c c}
\hline\hline
$n$ & $m$ & $m'$ & $k$ & $\lambda$ & $r$ & $|D|$ & $|C|$ & $\gamma$ & $T$\\[0.5ex]
\hline
1562 & 1408 & 1217733 & 1350 & 2000 & 3 &  1047 & 515 &  0.0267 & 10\\
1562 & 1408 & 1217733 & 1400 & 2000 & 4 &  1178 & 384 &  0.0484 & 11\\
1562 & 1408 & 1217733 & 1350 & 2000 & 5 &  1256 & 306 &  0.0833 & 11\\
1562 & 1408 & 1217733 & 1350 & 2000 & 6 &  1303 & 259 &  0.1091 & 12\\
\hline
\end{tabular}
\caption{Participating proteins: four runs on the complemented yeast protein network}
\label{table: PPI_qc}
\end{table}

\paragraph{Detection of functional modules:}
Some of the functional modules detected are shown in~{\bf Table~\ref{table: PPI_fm}}. Eventhough the density of the core $C$ of participating
proteins is small, the densities of the individual functional modules is very high. So, the Steps 1 and 2 have indeed helped in gathering
all the proteins more likely to participate in some functional module(s) and made it easier to filter away non-participating proteins.
Gene Ontology based verfication through~\cite{Yeast_GO} also shows high percentage of common Process, Function and Component ontologies
shared by the proteins involved in these modules. In this verification we submit the set of proteins belonging to a module onto the
GO browser and gather statistics about the shared properties. This also also shows several common functions are shared across modules.

\begin{table}[ht]
\scriptsize
\centering
\begin{tabular}{c | c c | c c c c}
\hline\hline
ID & Size & Density & \% shared & $p$-value & Description \\[0.5ex]
\hline
FM1 & 10 & 0.667 & 96\% & 1.27e-5 & Post-trans mod; cell division; phosphorylation\\
FM2 & 18 & 0.561 & 90\% & 9.75e-9 & Protein-kinase activity; receptor signaling; transferase activity\\
FM3 & 22 & 0.545 & 76\% & 8.32e-8 & Protein-kinase activity; MAP kinase; microfilament motor\\
FM4 & 18 & 0.543 & 80\% & 9.75e-9 & Protein-kinase activity; phosphotransferase activity; signal transducer\\
\hline
\end{tabular}
\caption{Functional modules verified through Yeast GO~\cite{Yeast_GO}}
\label{table: PPI_fm}
\end{table}

\paragraph{Discussions:}
The detection of individual functional modules in Step 3 is just one possible approach.
The information obtained through hubs and participating proteins can be used to detect functional modules and complexes using
other algorithms as well. For instance, the algorithm of~\cite{Guangyu} uses a heurestic search for finding cliques and `near cliques' 
on a protein network. This algorithm can be run on the subgraph formed by the participating proteins instead of the whole protein network considering 
a hub as the start node. Similarly, certain
clustering and random walk based methods given in~\cite{Vipin} which use random nodes as the initial start points can also make use
of our methods for finding the start nodes. Therefore, the hub and participating protein based information can be used in conjunction or as
filtering steps with other methods to arrive at both more efficient as well as accurate results.

From a technical point of view too we get a lot of insights into the problem from these experiments. 
We ran our algorithm on some standard, scale-free and random graphs. 
We noticed that $n$ and $|D|$ are almost linearly related. Also, $D$ is largest when
the slacks are zero, that is, vertex cover of the graph is the largest solution set for the~{\sc $(k,\lambda)$-lds} problem. 
As the slacks are increased, the size of $|D|$(and hence the parameter $k$) as well as the time $T$ to arrive at a feasible solution are reduced.

\section{Conclusions and future work}
\label{concl}
We discussed parameterized modeling and solutions to several practical problems in proteomics. Through this we showed how 
parameterized approaches can be viable alternatives to previously known sub-optimal techniques. We are in the process of biological
and statistical analysis of the experimental results. For instance, checking (a)~how many actual lethal proteins missed by the algorithm 
eventhough they were present in the dataset, (b)~how many non-hubs falsely detected as hubs, (c)~any new hub not present in public databases
but found in our dataset, etc. We are also in the process of experimenting on the cluster editing and quasi isomorphism problems on real data sets.
There is a lot of scope for related future research as well.
Firstly, one can look at more efficient parameterized algorithms for the same problems, especially
arriving at a linear kernel for the $(k,\lambda)$-LDS problem. Secondly, one can look at methods for graph partitioning and applying
our algorithms for each partition, thereby increasing the scalability of the approach. Thirdly, the current work can motivate researchers
to look at parameterized approaches to other problems in this field, for instance, protein function and structure prediction.


\end{document}